\documentclass[11pt]{article}
\usepackage{ifthen}
\usepackage{amsmath, amssymb}

\usepackage{amsthm}

\usepackage{ifpdf}

\usepackage{fullpage}
\usepackage[utf8]{inputenc}

\usepackage{algorithm,algpseudocode}
\algtext*{EndWhile} % Removes the "end while" text
\algtext*{EndIf} % Removes the "end if" text
\algtext*{EndFor} % Removes the "end for" text
\usepackage{amsfonts}
\usepackage{mathrsfs}
\usepackage{comment}
\usepackage{mathtools}

\usepackage{enumitem}

\usepackage{xcolor}

\usepackage{mdframed}

\usepackage{bbm}

\usepackage{hyperref}

\theoremstyle{plain}
\newtheorem{theorem}{Theorem}[section]
\newtheorem{lemma}[theorem]{Lemma}
\newtheorem{remark}[theorem]{Remark}
\newtheorem{corollary}[theorem]{Corollary}
\newtheorem{observation}[theorem]{Observation}
\theoremstyle{definition}

\begin{document}
	\title{Improved Deterministic Distributed Maximum Weight Independent Set Approximation in Sparse Graphs}
	\author{Yuval Gil\footnote{yuval.gil@campus.technion.ac.il}
	}
	\date{}
	
	\maketitle		
	
	\begin{abstract}
		We design new deterministic CONGEST approximation algorithms for \emph{maximum weight independent set (MWIS)} in \emph{sparse graphs}. As our main results, we obtain new $\Delta(1+\epsilon)$-approximation algorithms as well as algorithms whose approximation ratio depend strictly on $\alpha$, in graphs with maximum degree $\Delta$ and arboricity $\alpha$. For (deterministic) $\Delta(1+\epsilon)$-approximation, the current state-of-the-art is due to a recent breakthrough by Faour et al.\ [SODA 2023] that showed an $O(\log^{2} (\Delta W)\cdot \log (1/\epsilon)+\log ^{*}n)$-round algorithm, where $W$ is the largest node-weight (this bound translates to $O(\log^{2} n\cdot\log (1/\epsilon))$ under the common assumption that $W=\text{poly}(n)$). As for $\alpha$-dependent approximations, a deterministic CONGEST $(8(1+\epsilon)\cdot\alpha)$-approximation algorithm with runtime $O(\log^{3} n\cdot\log (1/\epsilon))$ can be derived by combining the aforementioned algorithm of Faour et al.\  with a method presented by Kawarabayashi et al.\ [DISC 2020].  As our main results, we show the following.
		\begin{itemize}
			\item A deterministic CONGEST algorithm that computes an $\alpha^{1+\tau}$-approximation for MWIS in $O(\log n\log \alpha)$ rounds for any constant $\tau> 0$. To the best of our knowledge, this is the fastest runtime of any deterministic \emph{non-trivial} approximation algorithm for MWIS to date. Furthermore, for the large class of graphs where $\alpha=\Delta^{1-\Theta(1)}$, it implies a deterministic $\Delta^{1-\Theta(1)}$-approximation algorithm with a runtime of $O(\log n\log \alpha)$ which improves upon the result of Faour et al.\ in both approximation ratio (by a $\Delta^{\Theta(1)}$ factor) and runtime (by an $O(\log n/\log \alpha)$ factor). 
			\item A deterministic CONGEST algorithm that computes an $O(\alpha)$-approximation for MWIS in $O(\alpha^{\tau}\log n)$ rounds for any (desirably small) constant $\tau>0$. This improves the runtime of the best known deterministic $O(\alpha)$-approximation algorithm in the case that $\alpha=O(\text{poly}\log n)$. This also leads to a deterministic $\Delta(1+\epsilon)$-approximation algorithm with a runtime of $O(\alpha^{\tau}\log n\log (1/\epsilon))$ which improves upon the runtime of Faour et al.\  in the case that $\alpha=O(\text{poly}\log n)$.
			\item A deterministic CONGEST algorithm that computes a $(\lfloor(2+\epsilon)\alpha\rfloor)$-approximation for MWIS in $O(\alpha\log n)$ rounds. This improves upon the best known $\alpha$-dependent approximation ratio by a constant factor.
			\item A deterministic CONGEST algorithm that computes a $2d^2$-approximation for MWIS in time $O(d^2+\log ^{*}n)$ in a directed graph with out-degree at most $d$. The dependency on $n$ is (asymptotically) optimal due to a lower bound by Czygrinow et al.\  [DISC 2008] and Lenzen and Wattenhofer  [DISC 2008].
		\end{itemize}
		We note that a key ingredient to all of our algorithms is a novel deterministic method that computes a high-weight subset of nodes whose induced subgraph is sparse.
	\end{abstract}
	%\renewcommand{\therequirement}{R\arabic{requirement}}

	%%%%%%%%%%%%%%%%%%%%%%%%%%%%%%%%%%%%%%%%%%%%%%%%%%%%%%%%%%%%%%%%%%%%%%%%%%%%%%
	\section{Introduction}
	\label{section:introduction}
	%%%%%%%%%%%%%%%%%%%%%%%%%%%%%%%%%%%%%%%%%%%%%%%%%%%%%%%%%%%%%%%%%%%%%%%%%%%%%%
	
	The problem of finding a \emph{maximum independent set (MaxIS)} in a graph is long known to be NP-hard. This is because a MaxIS is the complement set of a minimum vertex cover which is among the $21$ problems that appear in Karp's seminal work \cite{Karp72}. 
	
	Further research showed that assuming $NP\neq ZPP$, MaxIS cannot be approximated within $n^{1-\epsilon}$ efficiently for any constant $\epsilon >0$ \cite{Hastad96}. In graphs with maximum degree $\Delta$, a bound of $\Omega(\Delta /\log ^{2}\Delta)$ on the approximation ratio is shown in \cite{AustrinKS11} assuming the Unique Games Conjecture; whereas a bound of $\Omega(\Delta /\log ^{4}\Delta)$ is shown in \cite{Chan16} assuming that $P\neq NP$. 
	
	On the positive side, efficient approximation algorithms were developed by numerous researchers. An $O(n(\log\log n)^{2}/\log^3 n)$-approximation is presented in \cite{Feige04}. As for approximations depending on $\Delta$, various studies obtained an $O(\Delta\log\log \Delta/\log\Delta)$-approximation (some of which also apply for the weighted case) \cite{AlonK98,Halldorsson98,Halldorsson00,Halperin02,KargerMS98}.
	
	In this paper, we focus on distributed \emph{maximum weight independent set (MWIS)} approximation algorithms. The MWIS problem has been the subject of various studies in both the LOCAL and CONGEST models of distributed computing. 
	
	In the LOCAL model, a deterministic  $(1+\epsilon)$-approximation algorithm for planar graphs in $O(\log^{*}n)$ rounds was presented in \cite{CzygrinowHW08}. For general graphs, the authors of \cite{GhaffariKM17} obtained a randomized $(1+\epsilon)$-approximation algorithm in $\text{poly}\log n$ rounds based on network decomposition. This result was later derandomized in \cite{RozhonG20}. As for lower bounds, it is shown in \cite{CzygrinowHW08} and \cite{LenzenW08} that any deterministic $O(\log^{*}n)$-approximation for (unweighted) MaxIS requires $\Omega(\log^{*}n)$ rounds. This is extended in \cite{KawarabayashiKS20} to an $\Omega(\log ^{*}n)$ lower bound on the number of rounds required for a randomized algorithm to compute an independent set of size $\Omega(n/\Delta)$. 
	
	In the CONGEST model, the main result of \cite{Bar-YehudaCGS17} is a $\Delta$-approximation in $O(\text{MIS}(G)\cdot\log W)$ rounds, where $\text{MIS}(G)$ is the number of rounds required to find a maximal independent set (MIS) in graph $G$ and $W$ is the largest weight (commonly assumed to be polynomial in $n$). MWIS was further studied in \cite{KawarabayashiKS20}, where a randomized $\Delta(1+\epsilon)$-approximation in $\text{poly}(\log \log n)$ rounds is presented. Deterministic MWIS approximation was considered recently in \cite{FaourGGKR22}. Among other results, a deterministic $B(1+\epsilon)$-approximation in $O(\log^{2}(\Delta W)\log(\epsilon^{-1})+\log ^{*}n)$ rounds is presented, where $B$ is the neighborhood independence. On the negative side, hardness results for exact and approximate MaxIS computation in the CONGEST model are shown in \cite{Censor-HillelKP17,BachrachCDELP19,EfronGK20}.
	
	\subsection{Our Objective and Motivation} \label{section:motivation}
	In this paper, our goal is to study deterministic CONGEST algorithms for the MWIS problem in \emph{sparse graphs}. Distributed graph algorithms in sparse graphs have been a focal point of much research since they were considered in the seminal work of Goldberg et al.\ \cite{GoldbergPS88}. In  \cite{BarenboimE10}, Barenboim and Elkin presented fast deterministic algorithms for some classical symmetry-breaking problems on \emph{bounded arboricity} graphs. The arboricity of a graph, denoted by $\alpha$, is the minimum number of edge disjoint forests into which the edges of the graph can be partitioned. The rich class of bounded arboricity graphs include many well-studied and important graph families such as planar graphs, graphs of bounded treewidth,  and graphs excluding a fixed minor.
	
	In recent years, there is a growing interest in the class of bounded arboricity graphs (and its subclasses) in the context of \emph{distributed optimization}. As it turns out, it is often the case that bounded arboricity graphs allow for a better approximation and/or faster algorithms than in general graphs. Hence, a plethora of research has been devoted to distributed algorithms that operate on bounded arboricity graphs for various classical optimization problems such as maximum matching, minimum dominating set, and minimum vertex cover (see, e.g., \cite{abs-2305-01324,Dory0I22,CzygrinowHS09,CzygrinowHW08',CzygrinowHW08,CzygrinowHWW19,LenzenW10,MorganSW21,AmiriSS16,FaourFK21}). 
	
	The motivation for studying MWIS in bounded arboricity graphs is similar --- in the common case, the approximation ratio (and natural barrier) obtained by CONGEST algorithms is linear in $\Delta$ (which bounds $\alpha$ from above). Therefore, approximation algorithms that depend strictly on $\alpha$ are favorable for a large class of graphs. In light of that, Kawarabayashi et al.\ \cite{KawarabayashiKS20} showed that given a CONGEST $\Delta(1+\epsilon)$-approximation algorithm $\mathcal{A}$ with runtime $T(n,\Delta)$ (for $n$-node graph with maximum degree $\Delta$), there is a CONGEST $8(1+\epsilon)\alpha$-approximation algorithm for MWIS in time $O(T(n,O(\alpha))\cdot\log n )$.\footnote{We believe that the approximation ratio stated in \cite{KawarabayashiKS20} can be improved to $4(1+\epsilon)$ without affecting the (asymptotic) runtime by a slightly different choice of constants.} By the work of \cite{KawarabayashiKS20}, this leads to an $O(\log n \cdot \text{poly}(\log\log n)/\epsilon)$-round randomized algorithm. Regarding deterministic algorithms, plugging in the aforementioned result of \cite{FaourGGKR22} implies an $O(\log^{3} n\cdot \log(\epsilon^{-1}))$-round algorithm; whereas plugging in the deterministic $O(\Delta+\log^{*}n)$-round algorithm of \cite{Bar-YehudaCGS17} implies an $O(\log n\cdot (\alpha+\log ^{*}n))$-round algorithm. 
	
	%%%%%%%%%%%%%%%%%%%%%%%%%%%%%%%%%%%%%%%%%%%%%%%%%%%%%%%%%%%%%%%%%%%%%%%%%%%%%%
	\subsection{Our Contributions} \label{section:contributions}
	%%%%%%%%%%%%%%%%%%%%%%%%%%%%%%%%%%%%%%%%%%%%%%%%%%%%%%%%%%%%%%%%%%%%%%%%%%%%%%
	In this paper, we present new distributed deterministic approximation algorithms for the \emph{maximum weight independent set (MWIS)} problem. All of our algorithms operate in the CONGEST model (see Section \ref{section:preliminaries} for a definition). A key ingredient in our algorithms is a new procedure called $\mathtt{Sparse\_Set}$ which is outlined below.
	
	\subparagraph{Outline of Procedure $\mathtt{Sparse\_Set}$.} The goal of the $\mathtt{Sparse\_Set}$ procedure is simple: given a graph $G=(V,E)$ with node-weight function $w:V\rightarrow \mathbb{R}_{\geq 0}$, we seek to compute a subset $X\subseteq V$ with large weight $w(X)=\sum_{v\in X}w(v)$ whose induced subgraph $G(X)$ is relatively sparse. To achieve this goal, $\mathtt{Sparse\_Set}$ relies on a new notion called \emph{$\beta$-bounded coloring}. A $\beta$-bounded coloring is a proper node-coloring such that each node has at most $\beta$ neighbors with larger colors for some integer parameter $\beta>0$ (refer to Section \ref{section:preliminaries} for a formal definition). The $\mathtt{Sparse\_Set}$ procedure takes as input a $\beta$-bounded coloring $c$ of graph $G$ and an integer parameter $1\leq f\leq \beta$, and returns a subset $X\subseteq V$. The properties of $\mathtt{Sparse\_Set}$ are specified in the following lemma.
	
	\begin{lemma}\label{lemma:sparse-set-contributions}
		Let $G=(V,E)$ be a graph with node-weight function $w:V\rightarrow \mathbb{R}_{\geq 0}$, let $c$ be a $\beta$-bounded coloring of $G$ and let $1\leq f\leq \beta$ be an integer parameter. Upon termination, $\mathtt{Sparse\_Set}(c,f)$ returns a subset $X\subseteq V$ that satisfies: (1)  $|X\cap L(v)|< \beta/f$ for each selected node $v\in X$, where $L(v)=\{u\in N(v)\mid c(u)>c(v)\}$ is the set of $v$'s neighbors with larger color; (2) $ f\cdot w(X)\geq OPT(G)$, where $OPT(G)$ is the weight of a MWIS in $G$; and (3) $2\cdot f\cdot w(X)\geq w(V)$. The runtime of $\mathtt{Sparse\_Set}(c,f)$ is $O(k)$, where $k$ is the total number of distinct colors assigned by the coloring $c$.
	\end{lemma}
	Section \ref{section:sparse-set} is mostly dedicated to proving Lemma \ref{lemma:sparse-set-contributions}. This is done by means of a \emph{primal-dual} approach. 
	
	\subparagraph{Our bounds.} In Sections \ref{section:approx-arboricity} and \ref{section:maxis-approximation-directed}, we use $\mathtt{Sparse\_Set}$ in various ways to construct new approximation algorithms. Refer to Table \ref{table:comparison} for a list of our bounds.

	\begin{table}
		\begin{center}
			%\hspace*{-1.2 cm}
			\begin{tabular}{| l| l| l| }
				\hline
				\textbf{Approx.} & \textbf{Runtime}  & \textbf{Notes} \\
				
				\hline
				$\alpha^{1+\tau}$  & $O(\log n\log\alpha)$ & for any constant $\tau>0$
				\\
				\hline
				$\Delta^{1-\Theta(1)}$ & $O(\log n\log \alpha)$ &  restricted to graphs where $\alpha=\Delta^{1-\Theta(1)}$;\\
				&&  improves approx.\ ratio of \cite{FaourGGKR22} by a $\Delta^{\Theta (1)}$ factor;\\
				&& improves runtime of \cite{FaourGGKR22} by an $O(\log n/\log \alpha)$ factor
				\\
				\hline
				$O(\alpha)$ & $O(\alpha^{\tau} \log n)$& for any constant $\tau>0$;\\
				&&improves runtime of \cite{KawarabayashiKS20,FaourGGKR22} when $\alpha=O( \text{poly}\log n)$
				\\
				\hline
				$\Delta(1+\epsilon)$ & $O(\alpha^{\tau}  \log n\log (1/\epsilon))$ & for any constant $\tau>0$;\\
				&&improves runtime of \cite{FaourGGKR22} when $\alpha=O( \text{poly}\log n)$
				\\
				\hline
				$\lfloor (2+\epsilon)\cdot \alpha\rfloor$ & $O(\alpha \log n)$ & improves approx.\  ratio of \cite{KawarabayashiKS20}
				\\
				\hline
				$O(\alpha^{2})$ & $O(\log n+\sqrt {\alpha\log n}$ &\\
				& $+\alpha^{3/4}\log \alpha)$ &
				\\
				\hline
				$2d^{2}$ & $O(d^2+\log ^{*}n)$  &directed graphs with $\text{out-degree}\leq d$;\\ 
				&& $O(\log ^{*}n)$ is necessary due to \cite{CzygrinowHW08,LenzenW08}
				\\
				\hline
			\end{tabular}
		\end{center}
		\caption{A list of our MWIS approximations. Here, $\alpha$ denotes the arboricity and $\Delta$ denotes the maximum degree. Runtime improvements for $O(\alpha)$-approximation algorithm are compared with the deterministic $O(\log^{3}n)$-round algorithm (assuming $W=\text{poly}(n)$) derived from \cite{KawarabayashiKS20} and \cite{FaourGGKR22} (see Section \ref{section:motivation} for more details).}
		\label{table:comparison}
	\end{table}
	
	In Section  \ref{section:approx-arboricity}, we consider $\alpha$-arboricity graphs. First, using $\mathtt{Sparse\_Set}$ in a somewhat straightforward manner, we get the following theorem.
	\begin{theorem}\label{theorem:simple-arboricity-contributions}
		For any constant $\epsilon>0$, there exists a deterministic CONGEST algorithm that computes a $(\lfloor (2+\epsilon)\cdot \alpha\rfloor)$-approximation for MWIS in $O(\alpha\log n)$ rounds.
	\end{theorem}
	By a simple modification, we also establish the following theorem.
	\begin{theorem}\label{theorem:trade-off-alpha-contributions}
		There exists a deterministic CONGEST algorithm that computes an $O(\alpha^{2})$-approximation for MWIS in $O(\log n+\mathtt{COL}(n,(2+\epsilon)\cdot \alpha)+\sqrt {\alpha\log n})$ rounds, where $\mathtt{COL}(n,\Delta)$ is the runtime of $(\Delta+1)$-coloring an $n$-node graph with maximum degree $\Delta$.
	\end{theorem}
	For example, plugging in the $(\Delta+1)$-coloring algorithm of \cite{Barenboim16} leads to a runtime bound of  $O(\log n+\alpha^{3/4}\log \alpha+\sqrt {\alpha\log n})$.

	We then relax the $\lfloor (2+\epsilon)\cdot \alpha\rfloor$ approximation ratio of Theorem \ref{theorem:simple-arboricity-contributions} in favor of faster algorithms. Using $\mathtt{Sparse\_Set}$ in a slightly more elaborate way and pairing it with an arbdefective coloring (refer to Section \ref{section:preliminaries} for a definition) algorithm of \cite{BarenboimEG22}, we obtain the following lemma.
	
	\begin{lemma}\label{lemma:template-contributions}
		For any integer $k>0$, there exists a deterministic CONGEST algorithm that computes an $(8^{k}\cdot \alpha)$-approximation for MWIS in $O(k\cdot \alpha^{1/k} \cdot \log n)$ rounds.
	\end{lemma}
	As a consequence of Lemma \ref{lemma:template-contributions}, we obtain the following four theorems.
	\begin{theorem}\label{theorem:linear-in-alpha-contributions}
		For any constant $\tau>0$, there exists a deterministic CONGEST algorithm that computes an $O(\alpha)$-approximation for MWIS in $O(\alpha^{\tau}\log n)$ rounds.
	\end{theorem}
	\begin{theorem}\label{theorem:delta-epsilon-contributions}
		Let $\epsilon>0$ be a parameter. For any constant $\tau>0$, there exists a deterministic CONGEST algorithm that computes a $\Delta(1+\epsilon)$-approximation for MWIS in $O(\alpha^{\tau}\log n \log (1/\epsilon))$ rounds.
	\end{theorem}
	\begin{theorem}\label{theorem:one-plus-tau-contributions}
		For any constant $\tau>0$, there exists a deterministic CONGEST algorithm that computes an $\alpha^{1+\tau}$-approximation for MWIS in  $O(\log \alpha \log n)$ rounds.
	\end{theorem}
	\begin{theorem}\label{theorem:delta-conditional-contributions}
		For any graph $G=(V,E)$ with arboricity $\alpha$ and maximum degree $\Delta$ such that $\alpha=\Delta^{1-\Theta(1)}$, there exists a deterministic CONGEST algorithm that computes a $\Delta^{1-\Theta(1)}$-approximation for MWIS in $O(\log \alpha\log n)$ rounds.
	\end{theorem}
	We note that the time complexity stated in Theorems \ref{theorem:linear-in-alpha-contributions} and \ref{theorem:delta-epsilon-contributions} improve the state-of-the-art for approximations of $O(\alpha)$ and $\Delta(1+\epsilon)$, respectively, on graphs with $\alpha=O(\text{poly}\log n)$. Additionally, to the best of our knowledge, the $O(\log \alpha \log n)$ runtime of Theorems \ref{theorem:one-plus-tau-contributions} and \ref{theorem:delta-conditional-contributions} is the fastest known runtime of a deterministic non-trivial approximation algorithm for MWIS.
	
	Finally, in Section \ref{section:maxis-approximation-directed}, we consider directed graphs and obtain the following theorem.
	\begin{theorem}\label{theorem:directed-approximation-contributions}
		For a directed graph $G=(V,E)$ with out-degree at most $d$, there exists a deterministic CONGEST algorithm that computes a $2d^2$-approximation for MWIS in $O(d^2+\log ^{*}n)$ rounds.
	\end{theorem}
	We remark that the $O(\log ^{*}n)$ term in the time complexity is asymptotically tight due to an existing lower bound presented in \cite{CzygrinowHW08,LenzenW08} (refer to Section \ref{section:maxis-approximation-directed} for more details).

	%%%%%%%%%%%%%%%%%%%%%%%%%%%%%%%%%%%%%%%%%%%%%%%%%%%%%%%%%%%%%%%%%%%%%%%%%%%%%%
	\section{Preliminaries}
	\label{section:preliminaries}
	%%%%%%%%%%%%%%%%%%%%%%%%%%%%%%%%%%%%%%%%%%%%%%%%%%%%%%%%%%%%%%%%%%%%%%%%%%%%%%
	Consider a graph $G=(V,E)$ and denote $n=|V|$ and $m=|E|$. For each node $v\in V$, we denote by $N(v)$ the set of $v$'s \emph{neighbors} in $G$. Let $\deg(v)=|N(v)|$ be the \emph{degree} of node $v$ and let us denote by $\Delta=\max_{v\in V}\{\deg (v)\}$ the largest degree in the graph. If $G$ is directed, then we use the notation $(u\rightarrow v)$ to reflect that the edge $(u,v)\in E$ is directed from $u$ to $v$. In this context, the notation $(u,v)$ (or $(v,u)$) refers to an edge between $u$ and $v$ that could be directed in either direction. For a node $v\in V$, we say that a node $u\in N(v)$ is an \emph{incoming} (resp., \emph{outgoing}) neighbor of $v$ if there exists an edge $(u\rightarrow v)\in E$ (resp., $(v\rightarrow u)\in E)$). The \emph{in-degree} (resp., \emph{out-degree}) of node $v\in V$ is defined to be the number of $v$'s incoming (resp., outgoing) neighbors. For a node subset $U\subseteq V$, let $G(U)$ denote the subgraph induced by $U$. 
	
	\subparagraph*{The CONGEST model.}
	Our algorithms operate in the CONGEST model \cite{Peleg00}, where a communication network is abstracted by an $n$-node graph $G=(V,E)$. Each node $v\in V$ is equipped with a unique $O(\log n)$-bit identifier. Computation progresses in synchronous communication rounds, where each node $v\in V$ may send a message of size $O(\log n)$ to each neighbor $u\in N(v)$. If $G$ is directed, then the direction of each edge $(u,v)\in E$ is encoded to the endpoints $u$ and $v$ by means of a consistent orientation function. Notice that communication may occur on both directions of the edge $(u,v)$ regardless of its orientation.

	\subparagraph*{Maximum weight independent set.}
	Consider a graph $G=(V,E)$. A subset $X\subseteq V$ of nodes is said to be an \emph{independent set} of $G$ if it holds that  $(u,v)\notin E$ for all $u,v\in X$. For a node-weight function $w:V\rightarrow \mathbb{R}_{\geq 0}$, a \emph{maximum weight independent set (MWIS)} is an independent set $X\subseteq V$ that maximizes $w(X):=\sum_{v\in X}w(v)$.\footnote{Throughout this paper, we stick to the common assumption that all assigned weights can be represented using $O(\log n)$ bits and thus can be sent by means of  a single message in the CONGEST model.} We denote by $OPT(G)$ the weight of a MWIS in $G$. As usual, for a parameter $q\geq 1$, we say that a subset $X\subseteq V$ of nodes is a \emph{$q$-approximation} for MWIS if it is an independent set that satisfies $q\cdot w(X)\geq OPT(G)$.
	
	In a natural linear program (LP) relaxation of MWIS, each node $v\in V$ is associated with a variable $x_{v}$. The objective is to maximize $\sum_{v\in V}w(v)\cdot x_{v}$, subject to the constraints $x_{u}+x_{v}\leq 1$ for each edge $(u,v)\in E$; and $x_{v}\geq 0$ for each node $v\in V$. In the dual LP, each edge $e\in E$ is associated with a variable $y_{e}$. The dual objective is to minimize $\sum_{e\in E}y_{e}$, subject to the constraints $\sum_{u\in N(v)}y_{u,v}\geq w(v)$ for each node $v\in V$; and $y_{e}\geq 0$ for each edge $e\in E$. The \emph{weak duality} theorem \cite[Chapter 12, Theorem 12.2]{VaziraniBook} implies that $w(X)\leq \sum_{e\in E}y_{e}$ for any independent set $X\subseteq V$ and feasible dual solution $\mathbf{y}=\{y_{e}\}_{e\in E}$ (i.e., a dual solution that satisfies all constraints).

	\subparagraph*{$\beta$-Bounded coloring.} 
	Given a graph $G=(V,E)$ and integers $\beta,k>0$, we say that a node-coloring $c:V\rightarrow[k]$ is \emph{$\beta$-bounded} if the following conditions are satisfied: (1) $c$ is a \emph{proper} coloring, i.e., $c(u)\neq c(v)$ for all $(u,v)\in E$; and (2) each node $v\in V$, has at most $\beta$ neighbors with larger color, i.e., the set $L(v)=\{u\in N(v)\mid c(u)>c(v)\}$ satisfies $|L(v)|\leq \beta$ for all $v\in V$.\footnote{We sometimes naturally extend this definition to a coloring that assigns colors chosen from some ordered set of $k$ elements (and not necessarily $[k]$).}

	\subparagraph*{Arboricity.}
	Given an undirected graph $G=(V,E)$, the \emph{arboricity} of $G$ is defined to be the smallest integer $\alpha>0$ for which there exists a partition $E_{1},\dots ,E_{\alpha}$ of the edges into $\alpha$ pairwise-disjoint sets such that $(V,E_{i})$ is a forest for each $i\in [\alpha]$.
	
	\subparagraph*{Barenboim and Elkin's Partition Procedure.} In Section \ref{section:approx-arboricity} we make use of a partition procedure presented by Barenboim and Elkin in \cite{BarenboimE10}. This procedure takes as parameters the graph's arboricity $\alpha$ and a constant $\epsilon>0$. We shall refer to this procedure as $\mathtt{BE\_Partition}(\alpha,\epsilon)$. Procedure $\mathtt{BE\_Partition}(\alpha,\epsilon)$ partitions the nodes of graph $G$ into $\ell=O(\log n)$ layers $V=H_{1}\dot{\cup}\dots\dot{\cup}H_{\ell}$. For all $i\in[\ell]$, it is guaranteed that each node $v\in H_{i}$ has at most $\lfloor (2+\epsilon)\cdot \alpha\rfloor$ neighbors in $\cup_{j=i}^{\ell}H_j$. As established in \cite{BarenboimE10},  $\mathtt{BE\_Partition}$ takes $O(\log n)$ rounds in the CONGEST model.

	\subparagraph*{Arbdefective coloring.}
	The notion of \emph{$d$-arbdefective} coloring was introduced by Barenboim and Elkin in \cite{BarenboimE10'}. We say that a coloring $c:V\rightarrow[k]$ of graph $G=(V,E)$ is $d$-arbdefective if for every color $i\in [k]$, the subgraph $G(V_{i})$ induced by the subset $V_{i}=\{v\in V\mid c(v)=i\}$, has arboricity at most $d$. In \cite{BarenboimEG22}, it is shown that for any $1\leq p\leq \Delta$, a $(\Delta/p)$-arbdefective coloring that uses $O(p)$ colors can be computed in $O(p +\log ^{*}n)$ rounds in the CONGEST model.
	
	%%%%%%%%%%%%%%%%%%%%%%%%%%%%%%%%%%%%%%%%%%%%%%%%%%%%%%%%%%%%%%%%%%%%%%%%%%%%%%
	\section{The $\mathtt{Sparse\_Set}$ Procedure}
	\label{section:sparse-set}
	%%%%%%%%%%%%%%%%%%%%%%%%%%%%%%%%%%%%%%%%%%%%%%%%%%%%%%%%%%%%%%%%%%%%%%%%%%%%%%
	In this section, we present a simple procedure referred to as $\mathtt{Sparse\_Set}$ (Algorithm \ref{alg:sparsify}). Let $G=(V,E)$ be a graph with node-weight function $w:V\rightarrow \mathbb{R}_{\geq 0}$ and let $\beta\in \mathbb{Z}_{>0}$. The procedure takes as input a graph $G$, along with a $\beta$-bounded $k$-coloring $c$, and an integer parameter $f$ (encoded to each node) such that $1\leq f\leq \beta$, and returns a subset $X\subseteq V$ of \emph{selected} nodes. The properties of the set $X$ as well as the runtime of procedure $\mathtt{Sparse\_Set}$ are captured by the following lemma.
	
	\begin{lemma}\label{lemma:sparsify}
		Upon termination, $\mathtt{Sparse\_Set}(c,f)$ returns a subset $X\subseteq V$ that satisfies: (1)  $|X\cap L(v)|< \beta/f$ for each selected node $v\in X$, where $L(v)=\{u\in N(v)\mid c(u)>c(v)\}$ is the set of $v$'s neighbors with larger color; (2) $ f\cdot w(X)\geq OPT(G)$; and (3) $2\cdot f\cdot w(X)\geq w(V)$. The runtime of $\mathtt{Sparse\_Set}(c,f)$ is $O(k)$, where $k$ is the total number of distinct colors assigned by the coloring $c$.
	\end{lemma}
	
	We now describe the procedure $\mathtt{Sparse\_Set}$. Refer to Algorithm \ref{alg:sparsify} for a pseudocode description.

	\begin{algorithm}
		\scriptsize
		\caption{Procedure  $\mathtt{Sparse\_Set}(c,f)$ from the perspective of node $v\in V$.}
		\label{alg:sparsify}
		\hspace*{\algorithmicindent} \textbf{Input:} integer $1\leq f\leq \beta$;\\ \hspace*{\algorithmicindent} colors $c(v)$ and $c(u)$ for each $u\in N(v)$ ($c$ is a $\beta$-bounded coloring)
		\begin{algorithmic}[1]
			\State $v.y_{u,v}=\bot$ for all $u\in N(v)$ \Comment{$v$'s dual variables}
			\State $L(v)=\{u\in N(v)\mid c(u)>c(v)\}$ \Comment{$v$'s neighbors with larger color}
			\State $S(v)=N(v)- L(v)$ \Comment{$v$'s neighbors with smaller color}
			\State $IN(v)=OUT(v)=\emptyset$
			\State $\lambda (v)=\bot$
			\State $v.status=undecided$
			\While{$\lambda (v)==\bot$} \label{line:first-stage}\Comment{first stage}
			\For {each $u.y_{u,v}$ received from a neighbor $u\in S(v)$}
			\State $v.y_{u,v}=u.y_{u,v}$
			\EndFor
			\If{$v.y_{u,v}\neq \bot$ for all $u\in S(v)$}
			\State $\lambda (v)=\max\{0,w(v)-\sum_{u\in S(v)}v.y_{u,v}\}$
			\State $v.y_{u,v}=\frac{\lambda (v)\cdot f}{|L(v)|}$ for all $u\in L(v)$
			\State send $v.y_{u,v}$ to all $u\in L(v)$
			\If {$\lambda (v)==0$}
			\State $v.status=eliminated$
			\State send 'eliminated' to all $u\in S(v)$
			\EndIf
			\EndIf
			\label{line:first-stage-fin}
			\EndWhile
			\While{$v.status== undecided$ } \Comment{second stage}\label{line:second-stage}
			\For {each 'eliminated' received from a neighbor $u\in L(v)$}
			\State $OUT(v)=OUT(v)\cup\{u\}$ 
			\EndFor
			\For {each 'selected' received from a neighbor $u\in L(v)$}
			\State $IN(v)=IN(v)\cup\{u\}$ 
			\EndFor
			\If {$IN(v)\cup OUT(v)==L(v)$} \Comment{if all nodes in $L(v)$ are decided}
			\If {$|IN(v)|\geq \frac{|L(v)|}{f}$}
			\State $v.status=eliminated$
			\State send 'eliminated' to all $u\in S(v)$
			\Else
			\State $v.status=selected$
			\State send 'selected' to all $u\in S(v)$
			\EndIf
			\EndIf
			\label{line:second-stage-fin}
			\EndWhile
		\end{algorithmic}
	\end{algorithm}

	\subparagraph*{Overview of Algorithm \ref{alg:sparsify}.}
	Throughout the execution of procedure $\mathtt{Sparse\_Set}$, each node $v\in V$ maintains a status $v.status\in \{undecided,selected,eliminated\}$ and a dual variable $y_{u,v}$ for each neighbor $u\in N(v)$. We emphasize that the dual variables are used for analysis purpose. In particular, upon termination the dual solution $\mathbf{y}$ computed during procedure $\mathtt{Sparse\_Set}$ might be infeasible (we address that later in the analysis). Initially, each node sets its status to $undecided$ and all dual variables to $\bot$.  For each node $v\in V$, define $L(v)=\{u\in N(v)\mid c(u)>c(v)\}$ and $S(v)=N(v)-L(v)$ to be the sets of $v$'s neighbors with larger and smaller colors, respectively. Additionally, each node $v$ maintains a numerical value $\lambda(v)$ initialized to $\bot$.
	
	The execution of the algorithm can be divided into two subsequent stages: the first stage is dedicated to computing the dual variables (although some nodes may also be eliminated at this stage); whereas the goal of the second stage is for each undecided node to decide whether it becomes eliminated or selected.
	
	During the first stage (lines \ref{line:first-stage}--\ref{line:first-stage-fin}), each node $v\in V$ waits until it receives a dual variable $y_{u,v}$ from every neighbor $u\in S(v)$ (notice that if $S(v)=\emptyset$, then $v$ does not need to wait). Upon receiving all dual variables, $v$ sets $\lambda(v)=\max\{0, w(v)-\sum_{u\in S(v)}y_{u,v}\}$. Then, for each $u\in L(v)$, $v$ sets $y_{u,v}=\frac{\lambda (v)\cdot f}{|L(v)|}$ and informs $u$. If $\lambda(v)=0$, then $v$ becomes eliminated and informs all of its neighbors $u\in S(v)$.
	
	In the second stage (lines \ref{line:second-stage}--\ref{line:second-stage-fin}), all nodes become eliminated or selected. Each undecided node $v$ waits until it receives a message $\mu_{u}\in \{\text{'eliminated'},\text{'selected'}\}$ from every neighbor $u\in L(v)$. After receiving a message from every $u\in L(v)$, each undecided node $v$ becomes eliminated if at least $|L(v)|/f$ of its neighbors in $L(v)$ became selected; otherwise, $v$ becomes selected (notice that if $L(v)=\emptyset$, then $v$ becomes selected in the second stage if and only if it was not eliminated in the first stage). After changing its status, $v$ informs its neighbors $u\in S(v)$ of the new status. 
	
	We now analyze the procedure $\mathtt{Sparse\_Set}$. Let us denote by $X=\{v\in V\mid v.status=selected\}$ the set of nodes that were selected during the algorithm  and by $\mathbf{y}=\{y_{u,v}\mid (u,v)\in E\}$ the dual solution.

	Observe that the correctness of Property (1) of Lemma \ref{lemma:sparsify} follows directly from the construction of $\mathtt{Sparse\_Set}$ and the fact that $|L(v)|\leq \beta$. Hence, our goal now is to prove properties (2) and (3). To that end, for the sake of analysis, let us define a graph $G'=(V',E')$ as the graph obtained from $G$ by adding a virtual zero-weighted node $z$ and connecting it by edges to all nodes $v\in V$ with $L(v)=\emptyset$. That is, $G'$ consists of node set $V'=V\cup \{z\}$, edge set $E'=E\cup \{(z,v)\mid L(v)=\emptyset\}$, and weights $w'(z)=0$ and $w'(v)=w(v)$ for all $v\in V$. We also extend the definitions of $\mathbf{y}$ and the sets $L(v)$ to the graph $G'$ as follows. For each node $v\in V$ with $L(v)=\emptyset$ in $G$, we define $L'(v)=\{z\}$ and $y'_{v,z}=\lambda (v)\cdot f$; for nodes $v\in V$ with $L(v)\neq \emptyset$ in $G$, we keep $L'(v)=L(v)$ and $y'_{u,v}=y_{u,v}$ for each $u\in N(v)$. Let $N'(v)$ denote the set of $v$'s neighbors in $G'$ for each node $v\in V'$.

	\begin{observation}\label{observation:optg'}
		$OPT(G)= OPT(G')$.
	\end{observation}
	\begin{proof}
		Clearly, $OPT(G)\leq OPT(G')$ as any independent set of $G$ is also an independent set of $G'$. As for the other direction, consider a MWIS $I\subseteq V'$ of $G'$. The set $I\cap V$ is an independent set of $G$ with weight $w(I\cap V)=w(I)$. Thus, we get that $OPT(G)\geq w(I\cap V)=w(I)= OPT(G')$. 
	\end{proof}
	
	We note that Observation \ref{observation:optg'} implies that to establish Property (2) of Lemma \ref{lemma:sparsify}, it is sufficient to bound $w(X)$ in terms of $OPT(G')$. Furthermore, by weak duality, it is sufficient to bound $w(X)$ in terms of a feasible dual solution for $G'$. To that end, we show the following.
	
	\begin{observation}\label{observation:dualg'}
		$\mathbf{y}'$ is a feasible dual solution for $G'$.
	\end{observation}
	
	\begin{proof}
		First consider the virtual node $z$. Since $\mathbf{y}'$ is non-negative, it follows that $$\sum_{u\in N'(z)}y_{u,z}\geq w'(z)=0\ .$$ As for nodes $v\neq z$, we note that by construction, it follows that $$\sum_{u\in L'(v)}y_{u,v}=\sum_{u\in L'(v)}\frac{\lambda (v)\cdot f}{|L'(v)|}=\lambda (v)\cdot f\geq \lambda (v)\ .$$ Feasibility follows since $\lambda (v)=\max\{0,w(v)-\sum_{u\in S(v)}v.y_{u,v}\}\geq w(v)-\sum_{u\in S(v)}v.y_{u,v}$, which implies $\sum_{u\in N'(v)}y'_{u,v}=\sum_{u\in S(v)}y'_{u,v}+\sum_{u\in L'(v)}y'_{u,v}\geq w(v)$ for each node $v\in V$.
	\end{proof}
	
	We partition the nodes of $V-X$ into two sets $R_{1}=\{v\in V-X\mid \sum_{u\in S(v)}y'_{u,v}\geq w(v)\}$ and $R_{2}=V-X-R_{1}$. Observe that $R_{1}$ and $R_{2}$ are the sets of nodes that were eliminated in the first and second stage, respectively. 
	
	The following simple observation follows directly from the construction of $\mathbf{y}'$.
	\begin{observation}\label{observation:r1}
		Consider a node $v\in R_{1}$. For each $u\in L'(v)$, it holds that $y'_{u,v}=0$.
	\end{observation}
	From Observation \ref{observation:r1}, we can derive the following equality regarding the dual objective value.
	\begin{corollary}\label{corollary:r1isneg}
		$\sum_{e\in E'}y'_{e}=\sum_{v\in X\cup R_2}\sum_{u\in L'(v)}y'_{u,v}$
	\end{corollary}
	\begin{proof}
		First, notice that $\sum_{e\in E'}y'_{e}=\sum_{v\in V}\sum_{u\in L'(v)}y'_{u,v}$. The assertion follows since $$\sum_{v\in V}\sum_{u\in L'(v)}y'_{u,v}=\sum_{v\in X\cup R_2}\sum_{u\in L'(v)}y'_{u,v}+\sum_{v\in R_1}\sum_{u\in L'(v)}y'_{u,v}=\sum_{v\in X\cup R_2}\sum_{u\in L'(v)}y'_{u,v}$$
	\end{proof}
	We make the following observation regarding $R_{2}$.
	\begin{observation}\label{observation:r2}
		For each node $v\in R_{2}$, it holds that $|X\cap L'(v)|=|X\cap L(v)|\geq \frac{|L'(v)|}{f}$.
	\end{observation}
	\begin{proof}
		First, observe that nodes $v\in V$ with $L'(v)=\{z\}\neq L(v)$ are never eliminated in the second stage (they are either eliminated in the first stage or selected in the second stage). Hence, every node $v\in R_{2}$ satisfies $L(v)=L'(v)$ and thus $|X\cap L'(v)|=|X\cap L(v)|$. Now, for each node $v\in R_{2}$, the inequality $|X\cap L(v)|\geq \frac{|L'(v)|}{f}$ follows immediately from the fact that $v$ was eliminated in the second stage.
	\end{proof}
	
	We are now ready to bound $w(X)$ in terms of the objective value of $\mathbf{y}'$.
	\begin{lemma}\label{lemma:dual-bound}
		$f\cdot w(X)\geq\sum_{e\in E}y'_{e} $.
	\end{lemma} 
	\begin{proof}
		For each node $v\in R_{2}$, (arbitrarily) define a partition of $L'(v)$ into $|X\cap L'(v)|$ disjoint sets of at most $f$ nodes each. Observe that such partition exists since $f\cdot|X\cap L'(v)|\geq |L'(v)|$ by Observation \ref{observation:r2}. We now identify each such set with a node $u\in X\cap L'(v)$, and denote this set by $\mu_{v}(u)$. Now, for every $v\in X$, let $\rho(v)=\sum_{u\in S(v)\cap R_{2}}\sum_{u'\in \mu _{u}(v)}y'_{u,u'}$ and observe that by the construction of $\mathbf{y}'$, it holds that $\rho(v)= \sum_{u\in S(v)\cap R_{2}}|\mu _{u}(v)|\cdot \frac{\lambda (u)\cdot f}{|L'(u)|}\leq \sum_{u\in S(v)\cap R_{2}}f\cdot y'_{u,v}=f\cdot \sum_{u\in S(v)\cap R_{2}}y'_{u,v}$. 
		
		Recall that Corollary \ref{corollary:r1isneg} states that $\sum_{e\in E'}y'_{e}=\sum_{v\in X\cup R_2}\sum_{u\in L'(v)}y'_{u,v}$. Developing further, we get 
		\begin{align*}
			\sum_{v\in X\cup R_2}\sum_{u\in L'(v)}y'_{u,v}=\sum_{v\in X}\sum_{u\in L'(v)}y'_{u,v}+\sum_{v\in  R_2}\sum_{u\in L'(v)}y'_{u,v} &\, =\,  \\ 
			\sum_{v\in X}\sum_{u\in L'(v)}\frac{\lambda (v)\cdot f}{|L'(v)|}+\sum_{v\in  R_2}\sum_{u\in X\cap L'(v)}\sum_{u'\in \mu _{v}(u)}y'_{u',v}
			&\, =\,  \\
			\sum_{v\in X}\lambda (v)\cdot f+\sum_{v\in  X}\sum_{u\in S(v)\cap R_{2}}\sum_{u'\in \mu _{u}(v)}y'_{u',u}	&\, =\, \\ \sum_{v\in X}\lambda (v)\cdot f +\sum_{v\in X}\rho(v) &\,\leq \,\\ f\cdot \sum_{v\in X}\left(\lambda (v)+\sum_{u\in S(v)\cap R_{2}}y'_{u,v}\right)  &\,\leq \,\\ f\cdot \sum_{v\in X} \left(w(v)-\sum_{u\in S(v)\cap R_{2}}y'_{u,v}+\sum_{u\in S(v)\cap R_{2}}y'_{u,v}\right)  &\, =\,
			f\cdot w(X)\text{,}
		\end{align*}
		where the second sum in the second line holds because by definition, $L'(v)=\bigcup_{u\in X\cap L'(v)} \mu_{v}(u)$ for each $v\in R_{2}$, and the penultimate inequality holds because by construction, $\lambda(v)=w(v)-\sum_{u\in S(v)}y'_{u,v}\leq w(v)-\sum_{u\in S(v)\cap R_{2}}y'_{u,v}$ for each $v\in X$.
	\end{proof}
	
	It is now simple to show Property (3).
	\begin{lemma}[Property 3]\label{lemma:property3}
		$2\cdot f\cdot w(X)\geq w(V)$.
	\end{lemma}
	\begin{proof}
		It follows from Lemma \ref{lemma:dual-bound} that $f\cdot w(X)\geq \sum_{e\in E}y'_{e}$. Thus, it suffices to show that $2\cdot \sum_{e\in E}y'_{e}\geq w(V)$. Since $\mathbf{y}'$ is feasible, it holds that $\sum_{u\in N'(v)}y'_{u,v}\geq w(v)$ for each $v\in V'$. Therefore, we get $2\cdot \sum_{e\in E}y'_{e}=\sum_{v\in V'}\sum_{u\in N'(v)}y'_{u,v}\geq \sum_{v\in V'}w(v)=w(V')=w(V)$, where the first equality holds because in the summation $\sum_{v\in V'}\sum_{u\in N'(v)}y'_{u,v}$, we go over every edge twice.
	\end{proof}
	
	We are now prepared to prove Lemma \ref{lemma:sparsify}.
	\begin{proof}[Proof of Lemma \ref{lemma:sparsify}]
		As noted before, Property (1) follows directly from the $\mathtt{Sparse\_Set}$ construction. Property (2) follows from Observations \ref{observation:optg'} and \ref{observation:dualg'} and Lemma \ref{lemma:dual-bound}. Property (3) is established in Lemma \ref{lemma:property3}. 
		
		Regarding runtime, let us bound the number of rounds in each stage. By a simple inductive argument, it holds that after at most $i$ rounds of the first stage, each node $v$ colored by the $i$-th smallest color receives a dual variable $y_{u,v}$ from each neighbors $u\in S(v)$. Thus, the first stage finishes after at most $k$ rounds. Similarly, after at most $i$ rounds of the second stage, each node $v$ colored by the $i$-th largest color receives a message $\mu_{u}\in \{\text{'eliminated'},\text{'selected'}\}$ from every neighbor $u\in L(v)$. Hence, the second stage also finishes after at most $k$ rounds. Overall, we get that the total runtime is $O(k)$.
	\end{proof}
	
	We make another simple observation regarding the induced subgraph $G(X)$ that would be useful for the results obtained in later sections.
	\begin{observation}\label{observation:sparse-set-arb}
		The arboricity of $G(X)$ is smaller than $\beta/f$.
	\end{observation}
	\begin{proof}
		Let $\rho$ be the orientation obtained by directing each edge of $G(X)$ towards the endpoint that has a larger color according to $c$. Due to Property (1) of Lemma \ref{lemma:sparsify}, we get that every node $v\in X$ has less than $\beta/f$ outgoing edges in $G(X)$ under the orientation $\rho$. That is, the edges of $G(X)$ can be partitioned into $q<\beta/f$ subsets $E_{1},\dots, E_{q}$ of edges such that for each subset $E_{i}$, the out-degree of each node $v\in X$ is at most $1$. Moreover, by construction, each subgraph $(X,E_{i})$ is acyclic with respect to $\rho$. Thus, each subgraph $(X,E_{i})$ is a forest. Overall, we get that the edges of $G(X)$ can be partitioned into $q$ forests. 
	\end{proof}
	
	Notice that it follows from Lemma \ref{lemma:sparsify} that in the case that $\mathtt{Sparse\_Set}$ is invoked with $f=\beta$, the algorithm returns an independent set $X$ which is a $\beta$-approximation for MWIS.\footnote{We note that this special case of the $\mathtt{Sparse\_Set}$ procedure can be derived from the local-ratio approach presented in \cite{Bar-NoyBFNS01}.} That is, the following lemma is a special case of Lemma \ref{lemma:sparsify}.
	
	\begin{lemma}\label{lemma:beta-approximation}
		Given a $\beta$-bounded coloring $c$ of graph $G=(V,E)$, there exists an algorithm that computes a $\beta$-approximation $X\subseteq V$ for MWIS. Moreover, it holds that $2\cdot\beta\cdot w(X)\geq w(V)$. The runtime of this algorithm is $O(k)$, where $k$ is the total number of distinct colors assigned by $c$.
	\end{lemma}
	\begin{remark}\label{remark:knowledge-of-beta}
		In fact, to obtain the $\beta$-approximation of Lemma \ref{lemma:beta-approximation}, it suffices for each node $v\in V$ to replace $f$ with $|L(v)|$ (instead of with $\beta$). This modification does not affect the correctness nor the runtime. Hence, Lemma \ref{lemma:beta-approximation} can be accomplished even without the nodes knowing $\beta$.
	\end{remark}
	
	By some simple modifications, we are also able to devise a faster approximation algorithm (Algorithm \ref{alg:trade-off-beta}) while incurring a quadratic increase to the approximation guarantee of Lemma  \ref{lemma:beta-approximation}.  Specifically, we establish the following lemma.
	\begin{lemma}\label{lemma:trade-off}
		Given a $\beta$-bounded coloring $c:V\rightarrow [k]$, Algorithm \ref{alg:trade-off-beta} computes a $2\beta ^{2}$-approximation for MWIS. The runtime of Algorithm \ref{alg:trade-off-beta} is $O(\sqrt k)$.
	\end{lemma}
	The idea of Algorithm \ref{alg:trade-off-beta} is very simple. We divide the color $c(v)\in [k]$ of each node $v\in V$ to two colors $c_1(v),c_{2}(v)\in \{0,1,\dots,\lfloor\sqrt k\rfloor \}$ such that $c(v)=\lceil \sqrt k \rceil\cdot c_{1}(v)+c_{2}(v)$. Then, the algorithm computes a set $X\subseteq V$ by invoking the algorithm derived from Lemma \ref{lemma:beta-approximation} with the coloring $c_{1}$ on the graph $G_{1}=(V,E_{1})$, where $E_{1}=\{(u,v)\mid c_{1}(u)\neq c_{1}(v)\}$. Finally, to compute the output set $X'$, the algorithm of Lemma \ref{lemma:beta-approximation} is invoked again, this time with the coloring $c_{2}$ on the induced subgraph $G(X)$.
	\begin{algorithm}
		\caption{A $2\beta^{2}$-approximation algorithm for MWIS on a graph $G=(V,E)$ with given $\beta$-bounded coloring $c:V\rightarrow[k]$.}
		\label{alg:trade-off-beta}
		\begin{algorithmic}[1]
			\State let $(c_{1}(v),c_{2}(v))\in \{0,1,\dots,\lfloor\sqrt k\rfloor \}^{2}$ s.t.\ $c(v)=\lceil \sqrt k \rceil\cdot c_{1}(v)+c_{2}(v)$ for each $v\in V$
			\State let $E_1=\{(u,v)\mid c_{1}(u)\neq c_{1}(v)\}$ \Comment{bi-chromatic edges according to $c_{1}$}
			\State run $\mathtt{Sparse\_Set}(c_1,\beta)$ on $G_{1}=(V,E_{1})$ to obtain set $X\subseteq V$
			\State run $\mathtt{Sparse\_Set}(c_2,\beta)$ on $G(X)$ to obtain set $X'$
			\State return $X'$ as a MWIS approximation
		\end{algorithmic}
	\end{algorithm}
	
	We now prove Lemma \ref{lemma:trade-off}.
	\begin{proof}[Proof of Lemma \ref{lemma:trade-off}]
		Regarding runtime, we invoke $\mathtt{Sparse\_Set}$ with $c_{1}$ and $c_{2}$. Since both coloring functions use $O(\sqrt k)$ colors, the runtime complexity is $O(\sqrt k)$.
		
		As for correctness, we start by showing that $c_{1}$ is a $\beta$-bounded coloring on $G_{1}$. Notice that by definition, $c_{1}$ is proper with respect to $G_{1}$. To see that it is $\beta$-bounded, observe that $c_{1}(u)>c_{1}(v)\implies c(u)>c(v)$. Since $c$ is $\beta$-bounded, so is $c_{1}$. We can similarly show that $c_{2}$ is a $\beta$-bounded coloring with respect to $G(X)$. First, note that the edges of $G(X)$ must be monochromatic with respect to $c_{1}$. Since $c$ is a proper coloring, the edges of $G(X)$ must be bi-chromatic with respect to $c_{2}$. Moreover, it follows that $c_{2}(u)>c_{2}(v)\implies c(u)>c(v)$ which implies that $c_{2}$ is $\beta$-bounded in $G(X)$. The approximation ratio now follows since $2\beta \cdot w(X')\geq w(X)\geq OPT(G)/\beta \implies 2\beta^{2}\cdot w(X')\geq OPT(G)$.
	\end{proof}
	
	%%%%%%%%%%%%%%%%%%%%%%%%%%%%%%%%%%%%%%%%%%%%%%%%%%%%%%%%%%%%%%%%%%%%%%%%%%%%%%
	\section{Approximation Algorithms for $\alpha$-Arboricity Graphs}
	\label{section:approx-arboricity}
	%%%%%%%%%%%%%%%%%%%%%%%%%%%%%%%%%%%%%%%%%%%%%%%%%%%%%%%%%%%%%%%%%%%%%%%%%%%%%%
	In this section we present new deterministic approximation algorithms for MWIS on graphs $G=(V,E)$ with arboricity $\alpha$. For ease of presentation, we present our algorithms under the assumption that every node $v\in V$ knows the value of $\alpha$. Refer to Section \ref{section:knowledge} for a discussion on how this assumption can be lifted at the cost of a multiplicative factor of at most $O(\log \alpha)$ to the runtime.
	
	%%%%%%%%%%%%%%%%%%%%%%%%%%%%%%%%%%%%%%%%%%%%%%%%%%%%%%%%%%%%%%%%%%%%%%%%%%%%%%
	\subsection{A Basic $(\lfloor (2+\epsilon)\cdot \alpha\rfloor)$-Approximation Algorithm}
	\label{section:simple-arboricity}
	%%%%%%%%%%%%%%%%%%%%%%%%%%%%%%%%%%%%%%%%%%%%%%%%%%%%%%%%%%%%%%%%%%%%%%%%%%%%%%
	In this section, we present an algorithm that computes a $(\lfloor (2+\epsilon)\cdot \alpha\rfloor)$-approximation for MWIS. More concretely, we prove the following theorem.
	\begin{theorem}\label{theorem:simple-arboricity}
		For any constant $\epsilon>0$, Algorithm \ref{alg:bounded-coloring} computes a $(\lfloor (2+\epsilon)\cdot \alpha\rfloor)$-approximation $X\subseteq V$ for MWIS. Moreover, $2\cdot(\lfloor (2+\epsilon)\cdot \alpha\rfloor)\cdot w(X)\geq w(V)$. The runtime of Algorithm \ref{alg:bounded-coloring} is $O(\alpha \log n)$.
	\end{theorem}
	
	We now describe the algorithm. Moving forward, we shall use the notation $\delta=\lfloor(2+\epsilon)\cdot \alpha\rfloor$. Refer to Algorithm \ref{alg:bounded-coloring} for a pseudocode description.
	
	\begin{algorithm}
		\caption{A $\delta$-approximation algorithm for an $\alpha$-arboricity graph $G=(V,E)$ ($\delta=\lfloor (2+\epsilon)\cdot \alpha\rfloor$).}
		\label{alg:bounded-coloring}
		\begin{algorithmic}[1]
			\State run $\mathtt{BE\_Partition}(\alpha,\epsilon)$  \Comment{computes node-partition into layers $H_1,\dots ,H_{\ell}\subseteq V$}
			\State compute $(\delta+1)$-coloring $\varphi_{i}$ of each subgraph $G(H_i)$ in parallel
			\State  each node $v$ of layer $i\in [\ell]$ chooses color $c(v)=(i,\varphi_{i}(v))$
			\State run $\mathtt{Sparse\_Set}(c,\delta)$ on $G$  to obtain set $X\subseteq V$ \label{line:c-delta}
			\State return $X$ as a MWIS approximation
		\end{algorithmic}
	\end{algorithm}

	\subparagraph*{Overview of Algorithm \ref{alg:bounded-coloring}.} 
	Algorithm \ref{alg:bounded-coloring} is very simple. First, $\mathtt{BE\_Partition}(\alpha,\epsilon)$ (see Section \ref{section:preliminaries}) is invoked to obtain layers $H_{1},\dots ,H_{\ell}$ (recall that $\ell=O(\log n)$). Then, a $(\delta+1)$-coloring $\varphi_{i}:H_{i}\rightarrow [\delta+1]$ is computed in every node-induced subgraph $G(H_{i})$ in parallel. Each node $v\in H_{i}$ chooses the color $c(v)=(i,\varphi_{i}(v))$. Finally, $\mathtt{Sparse\_Set}(c,\delta)$ is invoked on $G$ to compute the MWIS approximation $X\subseteq V$, where the ordering of colors required for the notion of $\beta$-bounded coloring is defined naturally as $c(v)=(i,\varphi_{i}(v))>c(u)=(j,\varphi_{j}(u))\iff (i>j)\lor (i=j\land \varphi_{i}(v)>\varphi_{j}(u))$.
	
	We are now ready to prove Theorem \ref{theorem:simple-arboricity}.
	\begin{proof}[Proof of Theorem \ref{theorem:simple-arboricity}]
		We start by analyzing the runtime. First, $\mathtt{BE\_Partition}$ takes $O(\log n)$ time. For the $(\delta+1)$-coloring of all layers, notice that the maximum degree in each subgraph $G(H_{i})$ is $\delta$. Thus, we can employ an algorithm that computes a $(\Delta+1)$-coloring on a graph with maximum degree $\Delta$. By \cite{Barenboim16}, this can be done in time $O(\delta^{3/4}\log \delta+\log ^{*}n)=O(\alpha^{3/4}\log \alpha+\log ^{*}n)$. Finally, the number of colors assigned by the coloring $c$ is $O(\delta\log n)=O(\alpha \log n)$. Thus, the call to $\mathtt{Sparse\_Set}(c,\delta)$ takes $O(\alpha\log n)$ time. Overall, the runtime of Algorithm \ref{alg:bounded-coloring} is $O(\alpha\log n)$.
		
		Towards showing correctness, we shall show that $c$ is a $\delta$-bounded coloring. To that end, first observe that $c$ is a proper coloring. Indeed, for an edge $(u,v)\in E$, if $u$ and $v$ are not in the same layer, then clearly $c(u)\neq c(v)$. Otherwise, let $i$ be the index such that $u,v\in H_{i}$. By the correctness of the coloring performed on $G(H_{i})$, we get that $u$ and $v$ were given different colors $\varphi_{i}(u)\neq \varphi_{i}(v)$ and thus, $c(u)\neq c(v)$.
		
		Consider a node $v\in V$, let $i\in [\ell]$ be the index such that $v\in H_{i}$, and let $L(v)=\{u\in N(v)\mid c(u)>c(v)\}$. To see that $c$ is $\delta$-bounded, notice that by definition of $c$, it follows that $L(v)\subseteq\cup_{j=i}^{\ell}H_j$. Since $v$ has at most $\delta$ neighbors in $\cup_{j=i}^{\ell}H_j$, we get that $|L(v)|\leq \delta$. The correctness of Algorithm \ref{alg:bounded-coloring} now follows directly from Lemma \ref{lemma:beta-approximation}
	\end{proof}
	We observe that replacing the invocation of $\mathtt{Sparse\_Set}$ (line \ref{line:c-delta}) with a call to Algorithm \ref{alg:trade-off-beta} results in the following theorem.
	
	\begin{theorem}\label{theorem:trade-off-alpha}
		There exists an algorithm that computes an $O(\alpha^{2})$-approximation for MWIS in $O(\log n+\mathtt{COL}(n,(2+\epsilon)\cdot \alpha)+\sqrt {\alpha\log n})$ rounds, where $\mathtt{COL}(n,\Delta)$ is the runtime of $(\Delta+1)$-coloring an $n$-node graph with maximum degree $\Delta$.
	\end{theorem}
	For example, plugging the $(\Delta+1)$-coloring algorithm of \cite{Barenboim16} into Theorem \ref{theorem:trade-off-alpha} leads to a runtime bound of  $O(\log n+\alpha^{3/4}\log \alpha+\sqrt {\alpha\log n})$.

	%%%%%%%%%%%%%%%%%%%%%%%%%%%%%%%%%%%%%%%%%%%%%%%%%%%%%%%%%%%%%%%%%%%%%%%%%%%%%%
	\subsection{Faster Algorithms}
	%%%%%%%%%%%%%%%%%%%%%%%%%%%%%%%%%%%%%%%%%%%%%%%%%%%%%%%%%%%%%%%%%%%%%%%%%%%%%%
	In this section we present a generic approximation algorithm (Algorithm \ref{alg:template}) parameterized by an integer $k>0$. Specifically, we show the following.
	\begin{lemma}\label{lemma:template}
		For any integer $k>0$, Algorithm \ref{alg:template} computes an $(8^{k}\cdot \alpha)$-approximation for MWIS.  The runtime of Algorithm \ref{alg:template} is $O(k\cdot \alpha^{1/k} \cdot \log n)$.
	\end{lemma}
	Later on in the section, we demonstrate the applicability of the generic algorithm (see Theorems \ref{theorem:linear-in-alpha},\ref{theorem:delta-epsilon},\ref{theorem:one-plus-tau}, and \ref{theorem:delta-conditional}). We now give an overview of the algorithm. Refer to Algorithm \ref{alg:template} for a pseudocode description.
	\begin{algorithm} 
		\caption{A $(8^{k}\cdot \alpha)$-approximation algorithm for an $\alpha$-arboricity graph $G=(V,E)$ and parameter $k\in\mathbb{Z}_{\geq 1}$.}
		\label{alg:template}
		\begin{algorithmic}[1]
			\State $X_{0}=V$
			\For{$t=0,\dots, k-2$} 
			\State run $\mathtt{BE\_Partition}(\alpha^{(k-t)/k},\epsilon= 1/100)$ on $G(X_{t})$ to compute layers $H^{t}_{1},\dots,H^{t}_{\ell}$
			\State compute a $(\frac{1}{4}\alpha^{(k-t-1)/k})$-arbdefective $O(\alpha^{1/k})$-coloring  $\varphi_{i}$
			on each $G(H^{t}_i)$ \label{line:arbdefective-col}
			\State each node $v$ of layer $i\in [\ell]$ chooses the color $c_{t}(v)=(i,\varphi_{i}(v))$
			\State let $B_{t}=\{(u,v)\in E\mid u,v\in X_{t},\  c_{t}(u)\neq c_{t}(v)\}$
			\State run $\mathtt{Sparse\_Set}(c_{t},4\alpha^{1/k})$ on subgraph $(X_{t},B_{t})$  to obtain set $X_{t+1}\subseteq V$
			\EndFor
			\State run Algorithm \ref{alg:bounded-coloring} on $G(X_{k-1})$ with parameters $\alpha^{1/k}$ and $\epsilon=1/100$ to obtain set $X_{k}\subseteq V$ \label{line:final-call-template}
			\State return $X_{k}$ as a MWIS approximation
			
		\end{algorithmic}
	\end{algorithm}
	\subparagraph*{Overview of Algorithm \ref{alg:template}.}  
	The algorithm starts by performing $k-1$ phases $t=0,\dots ,k-2$. Starting from $X_{0}=V$, in each phase $t$, the set $X_{t+1}\subseteq X_{t}$ is computed. The goal is for $X_{t+1}$ to not be "too far" from $X_{t}$ in terms of weight while having a considerably sparser induced subgraph. To that end, we start by computing a coloring $c_{t}$ as follows. First, $\mathtt{BE\_Partition}$ is invoked on $G(X_{t})$ to compute layers $H^{t}_{1},\dots,H^{t}_{\ell}$. Then, an arbdefective coloring $\varphi_{i}$ that uses $O(\alpha^{1/k})$ colors is computed on each $G(H^{t}_i)$. The coloring $c_{t}$ is obtained by taking $c_{t}(v)=(i,\varphi_{i}(v))$ for each node $v\in V$. The set $X_{t+1}$ is then computed by an invocation of $\mathtt{Sparse\_Set}(c_{t},4\alpha^{1/k})$ on the subgraph $(X_{t},B_{t})$, where $B_{t}=\{(u,v)\in E\mid u,v\in X_{t},\  c_{t}(u)\neq c_{t}(v)\}$ is the set of bi-chromatic edges (with respect to $c_{t}$) in $G(X_{t})$. After completing the $k-1$ phases, the output set $X_{k}$ is computed by an invocation of Algorithm \ref{alg:bounded-coloring} on the subgraph $G(X_{k-1})$ (notice that in the case of $k=1$, Algorithm \ref{alg:template} is just an invocation of Algorithm \ref{alg:bounded-coloring}).
	
	We now analyze Algorithm \ref{alg:template} starting from the following observation.
	\begin{observation}\label{observation:c-and-alpha}
		Every $0\leq t\leq k-1$, satisfies the following:  (1) the arboricity of $G(X_{t})$ is at most $\alpha^{(k-t)/k}$; and (2) $c_{t}$ is a $(3\alpha^{(k-t)/k})$-bounded coloring of the subgraph $(X_t,B_{t})$.
	\end{observation}
	\begin{proof}
		We prove the assertion by induction over $t$. For the base case, consider $t=0$. By definition, the arboricity of $G(X_{0})=G$ is $\alpha=\alpha^{(k-0)/k}$. To see that $c_{0}$ is indeed a $(3\alpha)$-bounded coloring of $(X_{0},B_{0})$, first notice that by the definition of $B_{0}$, it holds that $c_{0}(u)\neq c_{0}(v)$ for every edge $(u,v)\in B_{0}$. That is, $c_{0}$ is a proper coloring of $(X_{0},B_{0})$. Moreover, by the properties of $\mathtt{BE\_Partition}$, it follows that each node $v\in H^{0}_{i}$ has at most $(2+1/100)\alpha<3\alpha$ neighbors in the layers $\cup_{j=i}^{\ell}H^{0}_j$. We can now establish that $c_{0}$ is  $(3\alpha)$-bounded since for each node $v\in H^{0}_{i}$, all of $v$'s neighbors with larger color must be in $\cup_{j=i}^{\ell}H^{0}_j$.
		
		Suppose now that the assertion holds for some $t\geq 0$. First, it is essential to show that the arbdefective coloring of line \ref{line:arbdefective-col} is in fact computable. Recall that by \cite{BarenboimEG22}, it is possible to compute a $(\Delta/p)$-arbdefective coloring that uses $O(p)$ colors in graphs of maximum degree $\Delta$. Notice that by the induction hypothesis, $G(X_{t})$ has arboricity at most $\alpha^{(k-t)/k}$. Hence, by the properties of $\mathtt{BE\_Partition}$, each $G(H^{t}_{i})$ has maximum degree $O(\alpha^{(k-t)/k})$. Now, for a fitting choice of $p=O(\alpha^{1/k})$, we get an arbdefective coloring with the desired parameters.
		
		For the step of the induction, we start by showing that the arboricity of $G(X_{t+1})$ is at most $\alpha^{(k-t-1)/k}$. Let $B_{t}=\{(u,v)\in E\mid u,v\in X_{t},\  c_{t}(u)\neq c_{t}(v)\}$ be the set of bi-chromatic edges (with respect to $c_{t}$) in $G(X_{t})$. We partition the edges of $G(X_{t+1})$ into two disjoint sets $M=\{(u,v)\in E-B_{t}\mid u,v\in X_{t+1}\}$ and $B=\{(u,v)\in B_{t}\mid u,v\in X_{t+1}\}$. 
		
		Observe that by the definition of arbdefective coloring, the arboricity of $(X_{t+1},M)$ is at most $\frac{1}{4}\alpha^{(k-t-1)/k}$. As for the subgraph $(X_{t+1},B)$, first notice that by the induction hypothesis, $c_{t}$ is a $(3\alpha^{(k-t)/k})$-bounded coloring of $(X_{t},B_{t})$. Now, because $X_{t+1}$ is obtained by an invocation of $\mathtt{Sparse\_Set}(c_{t},4\alpha^{1/k})$ on $(X_{t},B_{t})$, it follows from Observation \ref{observation:sparse-set-arb} that the arboricity of the graph $(X_{t+1},B)$ is bounded from above by $\frac {3\alpha^{(k-t)/k}}{4\alpha^{1/k}}=\frac{3}{4}\alpha^{(k-t-1)/k}$. Overall, we get that the arboricity of $G(X_{t+1})$ is at most the sum of the two arboricities which is bounded by  $\frac{1}{4}\alpha^{(k-t-1)/k}+\frac{3}{4}\alpha^{(k-t-1)/k}=\alpha^{(k-t-1)/k}$.
		
		We are left to show that $c_{t+1}$ is a $(3\alpha^{(k-t-1)/k})$-bounded coloring of $(X_{t+1},B_{t+1})$. As we have already shown, the arboricity of $G(X_{t+1})$ is at most $\alpha^{(k-t-1)/k}$. Hence, by similar arguments to the ones presented for the base of the induction, it follows that $c_{t+1}$ is a $(3\alpha^{(k-t-1)/k})$-bounded coloring of $(X_{t+1},B_{t+1})$.
	\end{proof}
	Towards bounding the approximation ratio, we make the following observation. 
	\begin{observation}\label{observation:Xt}
		For every $0\leq t\leq k-1$, it holds that $w(X_{t})\leq 8\alpha^{1/k}\cdot w(X_{t+1})$.
	\end{observation}
	\begin{proof}
		First, consider $t=k-1$. By Observation \ref{observation:c-and-alpha}, the arboricity of $G(X_{k-1})$ is at most $\alpha^{1/k}$. Since we compute the set $X_{k}$ by running Algorithm \ref{alg:bounded-coloring} on $G(X_{k-1})$, it follows from Theorem \ref{theorem:simple-arboricity} that $w(X_{k-1})\leq 2\cdot(2+1/100)\cdot\alpha^{1/k}\cdot w(X_k)<8\alpha^{1/k}\cdot w(X_k)$. Now, for $t<k-1$, notice that $X_{t+1}$ is obtained by an invocation of $\mathtt{Sparse\_Set}(c_{t},4\alpha^{1/k})$. Thus, it follows from Lemma \ref{lemma:sparsify} that $w(X_{t})\leq 2\cdot4\alpha^{1/k}\cdot w(X_{t+1})=8\alpha^{1/k}\cdot w(X_{t+1})$.
	\end{proof}
	As a consequence of Observation \ref{observation:Xt}, we get the following lemma.
	\begin{lemma}\label{lemma:opt-bound}
		$OPT(G)\leq 8^{k}\cdot \alpha \cdot w(X_{k})$.
	\end{lemma}
	\begin{proof}
		It follows from Observation \ref{observation:Xt} that $w(X_{0})\leq (8\alpha^{1/k})^k\cdot w(X_{k})=8^{k}\cdot\alpha\cdot w(X_{k})$. The assertion follows since $OPT(G)\leq w(V)=w(X_0)$.
	\end{proof}
	We are now prepared to prove Lemma \ref{lemma:template}.
	\begin{proof}[Proof of Lemma \ref{lemma:template}]
		The correctness follows directly from Lemma \ref{lemma:opt-bound}. Regarding the runtime, consider a phase $0\leq t\leq k-2$. The invocation of $\mathtt{BE\_Partition}$ in phase $t$ takes $O(\log n)$ rounds. Then, computing an arbdefective coloring that uses $O(\alpha^{1/k})$ colors can be done in $O(\alpha^{1/k}+\log ^{*} n)$ due to the algorithm of \cite{BarenboimEG22}. Notice that the coloring $c_{t}$ uses $O(\alpha^{1/k}\cdot \ell)=O(\alpha^{1/k}\cdot \log n)$ colors. Thus, the runtime of $\mathtt{Sparse\_Set}(c_{t},4\alpha^{1/k})$ on subgraph $(X_{t},B_{t})$ takes $O(\alpha^{1/k}\cdot \log n)$ rounds. This means that the total runtime of all $k-1$ phases is $O(k\cdot \alpha^{1/k}\cdot \log n)$. Finally, in line \ref{line:final-call-template}, Algorithm \ref{alg:bounded-coloring} is invoked on a graph with arboricity at most $\alpha^{1/k}$. As established in \ref{theorem:simple-arboricity}, this requires $O(\alpha^{1/k}\cdot \log n)$ rounds. Overall, we get that the runtime of Algorithm \ref{alg:template} is $O(k\cdot \alpha^{1/k}\cdot \log n)$.
	\end{proof}
	
	We turn to explore some new bounds that can be derived from Lemma \ref{lemma:template}.
	\begin{theorem}\label{theorem:linear-in-alpha}
		For any constant $\tau>0$, there exists an $O(\alpha^{\tau}\log n)$-round algorithm that computes an $O(\alpha)$-approximation for MWIS.
	\end{theorem}
	\begin{proof}
		The algorithm is obtained by running Algorithm \ref{alg:template} with $k=\lceil 1/\tau\rceil$. The approximation ratio achieved is $8^{k}\cdot\alpha=8^{\lceil 1/\tau\rceil}\cdot \alpha =O(\alpha)$. The runtime is $O(k\cdot\alpha^{1/k}\cdot \log n)=O(\alpha^{\tau} \log n)$.
	\end{proof}
	Observe that since $\alpha\leq \Delta$, the algorithm of Theorem \ref{theorem:linear-in-alpha} is also an $O(\Delta)$-approximation algorithm. As described in \cite[Lemma 4.6]{FaourGGKR22} (based on the local-ratio approach of \cite{KawarabayashiKS20}), given an $O(\Delta)$-approximation algorithm $\mathcal{A}$ for MWIS, one can improve the approximation factor to $\Delta(1+\epsilon)$ for any parameter $\epsilon>0$, at the cost of $O(\log (1/\epsilon))$  repetitions of $\mathcal{A}$. This leads to the following theorem.
	\begin{theorem}\label{theorem:delta-epsilon}
		Let $\epsilon>0$ be a parameter. For any constant $\tau>0$, there exists an $O(\alpha^{\tau}\log n \log (1/\epsilon))$-round algorithm that computes a $\Delta(1+\epsilon)$-approximation for MWIS.
	\end{theorem}
	
	Another direct consequence of Algorithm \ref{alg:template} is the following theorem.
	\begin{theorem}\label{theorem:one-plus-tau}
		For any constant $\tau>0$, there exists an $O(\log \alpha \log n)$-round algorithm that computes an $\alpha^{1+\tau}$-approximation for MWIS.
	\end{theorem}
	\begin{proof}
		The algorithm is obtained by running Algorithm \ref{alg:template} with $k=\lfloor \frac{\tau}{3}\log \alpha\rfloor$. The approximation ratio achieved is $8^{k}\cdot\alpha\leq 8^{\frac{\tau}{3}\log \alpha}\cdot\alpha=\alpha^{\tau}\cdot \alpha=\alpha^{1+\tau}$. The runtime is $O(k\cdot\alpha^{1/k}\cdot \log n)=O(\log \alpha \log n)\cdot 2^{O(1/\tau)}=O(\log \alpha \log n)$.
	\end{proof}
	
	The following theorem follows directly from Theorem \ref{theorem:one-plus-tau}.
	\begin{theorem}\label{theorem:delta-conditional}
		For any graph $G=(V,E)$ with arboricity $\alpha$ and maximum degree $\Delta$ such that $\alpha=\Delta^{1-\Theta(1)}$, there exists an algorithm that computes a $\Delta^{1-\Theta(1)}$-approximation for MWIS in $O(\log \alpha\log n)$ rounds.
	\end{theorem}
	
	%%%%%%%%%%%%%%%%%%%%%%%%%%%%%%%%%%%%%%%%%%%%%%%%%%%%%%%%%%%%%%%%%%%%%%%%%%%%%%
	\subsection{Lifting the Knowledge of Arboricity Assumption}
	\label{section:knowledge}
	%%%%%%%%%%%%%%%%%%%%%%%%%%%%%%%%%%%%%%%%%%%%%%%%%%%%%%%%%%%%%%%%%%%%%%%%%%%%%%
	Recall that the algorithms presented in Section \ref{section:approx-arboricity} assume that the nodes know the value of $\alpha$. In this section, we show how this assumption can be lifted. First, we address the procedure $\mathtt{BE\_Partition}$ that is assumed to take $\alpha$ as input. We describe a method for lifting the assumption on the knowledge of $\alpha$ while incurring an $O(\log \alpha)$ multiplicative overhead to the number of layers $\ell$ computed during the procedure.
	
	The implementation of Procedure $\mathtt{BE\_Partition}$ without knowledge of $\alpha$ is similar to the one presented in \cite{BarenboimE10} with a small additional modification. The idea of \cite{BarenboimE10} is to run $\lceil\log n\rceil+1$ parallel executions of $\mathtt{BE\_Partition}(2^{i},\epsilon)$ for $i=0,1,\dots,\lceil \log n\rceil$. Each node then chooses the layer to which it was clustered in the execution that minimizes $i$. This leads to a partition $H_{1},\dots H_{\ell}$ of the nodes into $\ell=O(\log \alpha\log n)$ layers such that each node $v\in V$ belonging to layer $H_{j}$ has at most $2\delta$ neighbors in $\cup_{k=j}^{\ell}H_{k}$ (where $\delta=\lfloor (2+\epsilon)\alpha\rfloor$). A problem that arises in this case is that the approximation ratio of Algorithm \ref{alg:bounded-coloring} becomes $2\delta$ as the coloring $c$ obtained by the algorithm is $2\delta$-bounded. 
	
	To avoid increase to the approximation ratio, one can slightly tweak the parameters used in the runs of $\mathtt{BE\_Partition}$. Let $\gamma,\epsilon'>0$ be two constants such that $(2+\epsilon')(1+\gamma)\leq 2+\epsilon $. Now, change the parallel runs to execute $\mathtt{BE\_Partition}((1+\gamma)^{i},\epsilon')$ for $i=0,\dots ,\lceil \log_{1+\gamma}n\rceil$. In the resulting partition $H_1,\dots,H_{\ell}$, each node $v\in H_{j}$ has at most $\delta$ neighbors in $\cup_{k=j}^{\ell}H_{k}$. Moreover, notice that each node $v\in V$ obtains an estimate $\alpha(v)$ which is bounded from above by $\alpha$. We note that in our algorithms, whenever $\alpha$ is used, each node internally can simply use the estimate $\alpha(v)$. This modification does not affect the correctness.
	
	As noted before, when invoking $\mathtt{BE\_Partition}$ without knowledge of $\alpha$, the number of layers becomes $\ell=O(\log \alpha\log n)$. As a consequence, the number of colors in each coloring computed during our algorithms increases by an $O(\log \alpha)$ factor. Since the runtime of $\mathtt{Sparse\_Set}$ depends on the number of colors, we get that apart from Theorem \ref{theorem:trade-off-alpha}, the upper bound on all runtimes stated in Section \ref{section:approx-arboricity} is multiplied by $O(\log \alpha)$; whereas the runtime of the $O(\alpha^{2})$-approximation algorithm stated in Theorem \ref{theorem:trade-off-alpha} changes to $O(\log n+\mathtt{COL}(n,(2+\epsilon)\cdot \alpha)+\sqrt{ \alpha\log\alpha\log n})$.

	%%%%%%%%%%%%%%%%%%%%%%%%%%%%%%%%%%%%%%%%%%%%%%%%%%%%%%%%%%%%%%%%%%%%%%%%%%%%%%
	\section{Approximation Algorithm for Directed Graphs with Bounded Out-Degree}
	\label{section:maxis-approximation-directed}
	%%%%%%%%%%%%%%%%%%%%%%%%%%%%%%%%%%%%%%%%%%%%%%%%%%%%%%%%%%%%%%%%%%%%%%%%%%%%%%
	Consider a directed graph $G=(V,E)$ where the out-degree of each node is bounded by an integer $d>0$, and let $w:V\rightarrow\mathbb{R}_{\geq 0}$ be a node-weight function. In this section, we present an algorithm (Algorithm \ref{alg:directed}) that computes a $2d^2$-approximation for MWIS in $G$. More concretely, we prove the following theorem. 
	\begin{theorem}\label{theorem:directed-approximation}
		For a directed graph $G=(V,E)$ with out-degree at most $d$, Algorithm \ref{alg:directed} computes a $2d^2$-approximation for MWIS in time $O(d^2+\log ^{*}n)$.
	\end{theorem}
	\begin{algorithm}
		\caption{A $2d^2$-approximation algorithm for MWIS on a directed graph $G=(V,E)$ with out-degree at most $d$.}
		\label{alg:directed}
		\begin{algorithmic}[1]
			\State compute a proper $O(d^2)$-coloring $c$
			\State let $E_{inc}=\{(v\rightarrow u)\in E\mid c(v)<c(u)\}$
			\State run  $\mathtt{Sparse\_Set}(c,d)$ on $G_{inc}=(V,E_{inc})$ to obtain set $X\subseteq V$
			\State run $\mathtt{Sparse\_Set}(-c,d)$ on $G(X)$ to obtain set $X'\subseteq V$ \label{line:-c}
			\State return $X'$ as a MWIS approximation
		\end{algorithmic}
	\end{algorithm}
	
	\subparagraph*{Overview of Algorithm \ref{alg:directed}.} First, Algorithm \ref{alg:directed} computes a proper coloring $c:V\rightarrow[k]$ of $G$, such that $k=O(d^2)$. Then, let $E_{inc}=\{(v\rightarrow u)\in E\mid c(v)<c(u)\}$ be the set of edges in $G$ directed towards the endpoint with larger color. Algorithm \ref{alg:directed} then invokes $\mathtt{Sparse\_Set}(c,d)$ on the graph $G_{inc}=(V,E_{inc})$ to compute a set $X$. Notice that $X$ is not necessarily an independent set in $G$. To correct that, Algorithm \ref{alg:directed} invokes $\mathtt{Sparse\_Set}(-c,d)$ again, this time on the node-induced subgraph $G(X)$ with the coloring $(-c)$ (defined simply as $-c(v)=(-1)\cdot c(v)$ for all $v\in V$) to obtain the set $X'$ which is returned as a MWIS approximation.
	
	We now analyze Algorithm \ref{alg:directed} starting from the following two simple observations.
	\begin{observation}\label{observation:directed-feasible}
		$X'$ is an independent set in $G$.
	\end{observation}	
	%\begin{proof}
	%	Consider an edge $e\in E$. Notice that $X'\subseteq X$, thus if $e$ has an endpoint which is not in $X$, then it will not be in $X'$. Now, suppose that both endpoints of $e$ are in $X$. This means that $e$ is included in the edges of $G(X)$. Since $X'$ is an independent set in $G(X)$, it follows that at least one of $e$'s endpoints will not be in $X'$.
	%\end{proof}
	\begin{observation}\label{observation:directed-bounded-coloring}
		$c$ is a $d$-bounded coloring of $G_{inc}$.
	\end{observation}
	\begin{proof}
		Consider a node $v$. By construction, each neighbor $u$ of $v$  in $G_{inc}$ that satisfies $c(u)>c(v)$, must be an outgoing neighbor. Since $v$ has at most $d$ outgoing neighbors, it follows that $c$ is a $d$-bounded coloring.
	\end{proof}
	We now establish the following lemma regarding the weight of $X$.
	\begin{lemma}\label{lemma:directed-opt-inc}
		$d\cdot w(X)\geq OPT(G)$
	\end{lemma}
	\begin{proof}
		Let $OPT(G_{inc})$ be the weight of a MWIS in $G_{inc}$. Observation \ref{observation:directed-bounded-coloring} combined with Lemma \ref{lemma:beta-approximation} implies that $d\cdot w(X)\geq OPT(G_{inc})$. To see that $OPT(G_{inc})\geq OPT(G)$, observe that any independent set in $G$ is also an independent set in $G_{inc}$. Therefore, we get that $d\cdot w(X)\geq OPT(G_{inc})\geq OPT(G)$.
	\end{proof}
	We move on to bound the weight of $X'$. To that end, we make the following observation.
	\begin{observation}\label{observation:directed(-c)}
		$(-c)$ is a $d$-bounded coloring of $G(X)$.
	\end{observation}
	\begin{proof}
		Consider a node $v\in X$ and let $L'(v)=\{u\in X\cap N(v)\mid -c(v)<-c(u)\}$ be the set of $v$'s neighbors in $G(X)$ that have a larger color assigned by the coloring $(-c)$. Notice that equivalently, $L'(v)=\{u\in X\cap N(v)\mid c(v)>c(u)\}$. Consider an incoming neighbor $u\in N(v)$ of $v$. If $c(v)>c(u)$, then $(u\rightarrow v)\in E_{inc}$. Since $X$ is an independent set of $G_{inc}$ and $v\in X$, it follows that $u\notin X$. Therefore, all nodes in $L'(v)$ are outgoing neighbors of $v$. Because $v$ has at most $d$ outgoing neighbors, we get that $|L'(v)|\leq d$ which concludes our proof.
	\end{proof}
	The following corollary follows directly from Observations \ref{observation:directed-bounded-coloring} and \ref{observation:directed(-c)}, and Lemma \ref{lemma:beta-approximation}.
	\begin{corollary}\label{corollary:approximation}
		$2d\cdot w(X')\geq w(X)\geq OPT(G)/d$
	\end{corollary}

	We can now prove Theorem \ref{theorem:directed-approximation}.
	\begin{proof}[Proof of Theorem \ref{theorem:directed-approximation}]
		The correctness of Algorithm \ref{alg:directed} follows directly from Observation \ref{observation:directed-feasible} and Corollary \ref{corollary:approximation}. Regarding runtime, the coloring $c$ can be computed in $O(\log ^{*}n)$ rounds by means of a coloring algorithm by Linial \cite{Linial92}.\footnote{While Linial's algorithm computes an $O(\Delta^2)$-coloring in undirected graphs with maximum degree $\Delta$, it can be applied to compute an $O(d^{2})$-coloring in directed graph with out-degree at most $d$ in a straightforward manner.} Following that, $c$ and $(-c)$ use $O(d^2)$ colors. Thus, each invocation of $\mathtt{Sparse\_Set}$ takes $O(d^2)$ rounds. Overall, we get that the runtime of Algorithm \ref{alg:directed} is $O(d^2+\log ^{*}n)$.
	\end{proof}
	\subparagraph*{MWIS approximation requires $O(\log^{*}n)$ rounds.}
	In \cite{CzygrinowHW08},  Czygrinow et al.\ prove the following lemma.\footnote{A similar bound is presented in \cite{LenzenW08} by Lenzen and Wattenhofer.}
	\begin{lemma}[\cite{CzygrinowHW08}]
		There is no deterministic distributed algorithm that finds an independent set of size $\Omega(n/\log ^{*}n)$ in a cycle on n vertices in $o(\log ^{*} n)$ rounds.
	\end{lemma}
	We note that this lemma holds even in the LOCAL model, where the message size is not restricted. Moreover, the proof of Czygrinow et al.\ also applies for the case of oriented rings (where the edges are directed such that each node has out-degree $1$). Since any $n$-node ring contains an independent set of size $\Omega(n)$, this lower bound implies that any algorithm that computes an $O(\log ^{*}n)$-approximation for (unweighted) MaxIS in an oriented ring requires $\Omega(\log^{*}n)$ rounds. Therefore, the dependency on $n$ in the runtime of Algorithm \ref{alg:directed} cannot be improved.

	%%%%%%%%%%%%%%%%%%%%%%%%%%%%%%%%%%%%%%%%%%%%%%%%%%%%%%%%%%%%%%%%%%%%%%%%%%%%%%
	\clearpage
	\bibliographystyle{alpha}
	% the mandatory bibstyle
	
	\bibliography{references}
	%%%%%%%%%%%%%%%%%%%%%%%%%%%%%%%%%%%%%%%%%%%%%%%%%%%%%%%%%%%%%%%%%%%%%%%%%%%%%%
\end{document}